\documentclass[journal,comsoc]{IEEEtran}
\usepackage[T1]{fontenc}
\usepackage{graphicx}

\usepackage{algorithmic}
\usepackage{algorithm}
\usepackage{amsmath}
\usepackage{amssymb}
\usepackage{textcomp}
\usepackage{amsthm}
\usepackage{wrapfig}
\usepackage{lipsum}
\usepackage{xcolor}
\graphicspath{ {images/} }
\usepackage{array}
\ifCLASSOPTIONcompsoc
  \usepackage[caption=false,font=normalsize,labelfont=sf,textfont=sf]{subfig}
\else
  \usepackage[caption=false,font=footnotesize]{subfig}
\fi
\usepackage{stfloats}
\usepackage{url}
\usepackage{balance}
\newtheorem{theorem}{Theorem}

\usepackage[cmintegrals]{newtxmath}

\usepackage{cite}
\usepackage{amsmath,amssymb,amsfonts}
\usepackage{algorithmic}
\usepackage{graphicx}
\usepackage{textcomp}
\usepackage{xcolor}

\usepackage{tikz}
\usetikzlibrary{positioning}
\usetikzlibrary{shapes}
\usetikzlibrary{shapes.geometric, arrows}
\usepackage{graphicx}

\tikzstyle{startstop} = [circle, minimum width=2cm, minimum height=1cm,text centered, text width=1cm, draw=black, fill=red!30]
\tikzstyle{process} = [rectangle, rounded corners, minimum width=0.5cm, minimum height=1cm, text centered, text width=3cm, draw=black, fill=orange!30]
\tikzstyle{decision} = [diamond, minimum width=2.5cm, minimum height=0.5cm, text centered, text width=2.5cm, draw=black, fill=green!30]
\tikzstyle{arrow} = [thick,->,>=stealth]

\begin{document}

\title{Multivariate Extreme Value Theory Based Rate Selection for Ultra-Reliable Communications\\
}

\author{Niloofar~Mehrnia,~\IEEEmembership{Member,~IEEE,}
        Sinem~Coleri,~\IEEEmembership{Fellow,~IEEE}%

\thanks{Niloofar Mehrnia and Sinem Coleri are with the Department of Electrical and Electronics Engineering, Koc University, Istanbul, Turkey (e-mail: nmehrnia17@ku.edu.tr; scoleri@ku.edu.tr).}
\thanks{Niloofar Mehrnia is also with Koc University Ford Otosan Automotive Technologies Laboratory (KUFOTAL), Sariyer, Istanbul, Turkey.}
\thanks{Sinem Coleri acknowledges the support of Ford Otosan, and the Scientific and Technological Research Council of Turkey 2247-A National Leaders Research Grant $\#121C314$.}
}

\maketitle

\begin{abstract}
Diversity schemes play a vital role in improving the performance of ultra-reliable communication (URC) systems by transmitting over two or more communication channels to combat fading and co-channel interference.
Determining an appropriate transmission strategy that satisfies the ultra-reliability constraint necessitates the derivation of the statistics of the channel in the ultra-reliable region and, subsequently, integration of these statistics into the rate selection while incorporating a confidence interval to account for potential uncertainties that may arise during estimation. In this paper, we propose a novel framework for ultra-reliable real-time transmission considering both spatial diversities and ultra-reliable channel statistics based on multivariate extreme value theory (MEVT). First, the tail distribution of the joint received power sequences obtained from different receivers is modeled while incorporating the inter-relations of extreme events occurring rarely based on the Poisson point process approach in MEVT. The optimum transmission strategies are then developed by determining the optimum transmission rate based on the estimated joint tail distribution and incorporating confidence intervals (CIs) into the estimations to cope with the availability of limited data. Finally, the system reliability is assessed by utilizing the outage probability metric. 
Through analysis of the data obtained from the engine compartment of the Fiat Linea, our study showcases the effectiveness of the proposed methodology in surpassing traditional extrapolation-based approaches. This innovative method not only achieves a higher transmission rate, but also effectively addresses the stringent requirements of ultra-reliability.
The findings indicate that the proposed rate selection framework offers a viable solution for achieving a desired target error probability by employing a higher transmission rate and reducing the amount of training data compared to the conventional rate selection methods.
\end{abstract}


\section{Introduction}
\label{sec:intro}
Ultra-reliable communication (URC) plays a critical role in fifth-generation (5G) and beyond networks, enabling mission-critical applications in transportation, manufacturing, and healthcare. Examples of these applications include remote control of robots, remote surgery, wireless sensor networks within vehicles, autonomous vehicles, and vehicular teleoperation \cite{interface_01}-\nocite{5G_01}\nocite{interface_02}\nocite{urllc_01}\cite{urllc_02}. In current 5G standards, ultra-reliability and low latency are typically combined to form ultra-reliable low latency communication (URLLC). However, there are specific scenarios, such as health monitoring or disaster recovery, where the highest level of reliability is essential, but a slightly higher latency of more than the conventional $1$ millisecond (ms) is acceptable \cite{urllc_02}, \cite{urllc_05}. To achieve the ultra-reliability target of $10^{-9}$ to $10^{-5}$ packet error rate (PER), diversity schemes including time diversity \cite{timediv}, frequency diversity \cite{frequencydiv}, spatial diversity \cite{spacediv}, and interface diversity \cite{interferencediv} have been proposed. These schemes aim to reduce the required signal-to-noise ratio (SNR) to meet the stringent reliability constraints. Spatial diversity, which involves using multiple transmit and/or receive antennas to achieve reliable communication over fading channels, is preferred over frequency and interface diversities due to its cost-effectiveness and minimal transmission delay compared to time diversity \cite{urllc_diversity_01}. However, implementing ultra-reliable communication relies on accurately modeling extreme quantiles that occur infrequently. Therefore, significant advancements in statistical modeling of multiple channels and the development of corresponding transmission strategies are necessary for the successful utilization of spatial diversity in designing ultra-reliable communication systems.

In a URC system with spatial diversity, a proper channel model can be derived by i) estimating the channel statistic within a more nuanced notation of coherence 
distance for ultra-reliable communication \cite{urllc_08}, ii) using machine learning data-driven frameworks \cite{reliability_01}, or iii) incorporating novel techniques to the characteristics associated with extreme events across multiple variables, such as proposing the utilization of power law expression techniques by extrapolating from channel data \cite{urllc_02} and multivariate extreme value theory (MEVT) \cite{Mehrnia_twcBGPD}.
In \cite{urllc_08}, adjustments to the coherence distance definitions are made to accurately capture the distance at which wireless channels can be reliably predicted. The authors employ a quasi-static Rayleigh fading model, assuming spatial independence, to study channel dynamics in the context of URC.
Furthermore, \cite{reliability_01} addresses the estimation of channel-blocking incidents with probabilities on the order of $10^{-9}$-$10^{-5}$ using data-driven methods, without prior knowledge of dynamic channel statistics. This allows for the assessment of wireless connectivity reliability in ultra-reliable communication systems. However, machine learning data-driven approaches require a significant amount of training samples, typically exceeding a factor of ($10 \times \epsilon^{-1}$) to achieve the desired error probability $\epsilon$.
Additionally, \cite{urllc_02} proposes a power law framework that extends traditional channel models into the ultra-reliable region through extrapolation. This framework can be further enhanced to accommodate multiple receivers by employing a simplified expression of maximum ratio combining (MRC).
Nevertheless, these initial studies still rely on average statistics channels, such as Gaussian, Rayleigh, or Rician, to characterize channel behavior in the ultra-reliable region. Only recently, a novel methodology utilizing MEVT has been introduced in \cite{Mehrnia_twcBGPD} to model multiple channels in a multiple-input multiple-output (MIMO)-URC system by deriving the lower tail statistics of received power. MEVT is a statistical discipline specializing in formulating methods for modeling dependencies among infrequent events by leveraging multidimensional boundary relationships.
Despite numerous studies in the field of channel modeling, none have proposed a transmission strategy that relies on real-time assessment of the lower tail statistics of received power across multiple diversities.

The design of communication systems for ultra-reliable communications has been addressed in the context of adaptive relay selection techniques \cite{urllc_07}, and rate selection frameworks \cite{urllc_05}, \cite{reliability_03}, \cite{MehrniaTVTRate}. A cooperative communication-based relaying scheme has been proposed in \cite{urllc_07} to provide ultra-high reliability while meeting the latency requirements at a moderate signal-to-noise ratio (SNR) by modifying the definition of traditional coherence time and revising fading dynamics of wireless channels in the context of ultra-high reliability. On the other hand, \cite{urllc_05} uses a rate selection framework based on the extrapolation of average statistic-based channel models to design an ultra-reliable system. Moreover, \cite{reliability_03} develops a data-driven rate selection framework where federated learning (FL) training has been used to estimate the transmission rate and assess the reliability of the system with minimum training time in the absence of knowledge about the channel state information (CSI). However, all of these designs use the extrapolation of average-statistics channel models, which have been demonstrated not to fit the empirical data in the ultra-reliable region \cite{Mehrnia_twcBGPD}, \cite{MehrniaTVTRate}-\nocite{MehrniaTWC}\cite{MehrniaTVT}. 
We have recently proposed a novel extreme value theory (EVT)-based framework in \cite{MehrniaTVTRate} for estimating and validating the optimal transmission rate to address the constraints of ultra-reliable communications for a single transmitter-receiver pair. Nevertheless, none of these studies incorporate the statistics of multiple channels in the ultra-reliable region based on real data into the design of communication systems.

The estimation accuracy of the channel parameters and the transmission strategies need to be derived in the ultra-reliable regime due to the limited amount of training data by using confidence interval (CI) \cite{reliability_01}, \cite{confidenceinterval_03}, \cite{MehrniaEucnc} and bootstrapping techniques \cite{MehrniaTWC}. In \cite{reliability_01}, the authors propose a multi-layer perceptron (MLP) neural network (NN)-based machine learning data-driven framework incorporating CI to the statistics of a non-blocking connectivity duration to maintain the URLLC with the penalty of large training samples to satisfy the target error probability $\epsilon$. 
In \cite{confidenceinterval_03}, the authors evaluate the outage probability within a confidence level to address the reliability of the URC system by incorporating a non-parametric statistical learning algorithm into a rate selection framework assuming extrapolation of the average statistics fading channels. Only recently, in \cite{MehrniaEucnc}, we proposed a pioneering framework based on EVT that estimates the optimal transmission rate with a high degree of confidence using a reduced number of training samples. The proposed approach was further validated by evaluating the outage probability for ultra-reliable communications in a single transmitter-receiver scenario. 
On the other hand, in \cite{MehrniaTWC}, we develop a bootstrapping-based algorithm for the determination of the stopping conditions in collecting sufficient samples to estimate the tail statistics based on the normality assumptions of the return levels. Accordingly, if the sample size is sufficiently large to estimate the channel tail distribution of values exceeding a predetermined threshold by using the generalized Pareto distribution (GPD), the corresponding return levels obtained in the bootstrap iterations are normally distributed.
However, none of these studies incorporate ultra-reliable communication statistics of the diversities into the accuracy of simultaneously estimating the channel parameters and resulting transmission strategies.

The main objective of this study is to propose a transmission strategy for a real-time URC system that employs spatial diversity by considering the MEVT-based statistics of multiple channels while incorporating confidence intervals into parameter estimation to cope with the limited availability of data. The spotlight is on two-dimensional or bi-variate instances, with the aim of illustrating key concepts and concerns related to MEVT while avoiding the added complexity associated with a comprehensive multivariate approach. However, this methodology can be easily generalized to multiple dimensions.
According to the proposed methodology, the parameters of multiple channels are first estimated based on MEVT. Then, the optimal transmission rate is determined by incorporating the estimated multiple channel parameters while considering the ultra-reliability constraint subject to the target error probability $\epsilon$. Finally, the reliability of a system operating at the determined transmission rate is assessed by employing the outage probability. The uncertainty of the transmission rate is also obtained in relation to the confidence intervals of the estimated channel parameters. The original contributions of the paper are listed as follows:

\begin{itemize}
    \item We propose a novel framework for ultra-reliable real-time transmission considering both spatial diversities and ultra-reliable channel statistics based on MEVT, for the first time in the literature. The framework includes the modeling of the inter-relationship of the tail distributions of multiple channels by using the bi-variate GPD (BGPD) based on the Poisson point process approach, determination of the optimum transmission rate by using the estimated BGPD model, incorporating confidence intervals to the estimated transmission rate due to the restricted amount of data, and then assessment of the system reliability by utilizing the outage probability metric. 
    \item We propose a novel algorithm for the rate selection in a multi-diversity system based on the Poisson point process approach, for the first time in the literature. The algorithm incorporates the parameters of the probability measure function and BGPD into the optimum transmission rate to satisfy the ultra-reliability constraints.
    \item We propose an algorithm for the determination of the confidence interval of the transmission rate based on the confidence intervals of the estimated Uni-variate GPD (UGPD) parameters, for the first time in the literature. We determine the confidence intervals of the UGPD parameters by incorporating the bootstrapping-based bias-corrected accelerated (BCA) method that offers an order of magnitude improvement in accuracy for the parameter estimation using the maximum likelihood estimator (MLE). 
    \item We utilize a combination of one omnidirectional transmitter and two directional receivers to capture the data within the engine compartment of a Fiat Linea vehicle under different driving scenarios and engine vibrations. Through our proposed methodology, we achieve superior modeling accuracy for URC compared to conventional extrapolation-based models that rely on average statistical channels for multi-channel modeling.
    
\end{itemize}

The rest of the paper is organized as follows: Section \ref{sec:system_model} describes the system model and assumptions considered throughout the paper. Section~\ref{sec:framework} presents the MEVT-based rate selection framework for determining the optimum transmission rate and its affiliated confidence interval for a system operating in URC. Section \ref{sec:numerical_results} provides the channel measurement setup and the performance evaluation in determining the optimum transmission rate and outage probability, as well as the confidence interval for the estimated rate. Finally, concluding remarks and future works are given in Section \ref{sec:conclusions}.

\section{System Model}
\label{sec:system_model}
We consider a one-way communication in which a transmitter (Tx) sends data packets to two receivers, namely Rx$1$ and Rx$2$, at the total rate $R$, addressing the requirements of ultra-reliable communication. The framework can be easily extended for the scenario in which more than two receivers exist without loss of generality. Assuming a fixed and predetermined transmit power, estimating the received signal power is tantamount to determining the squared amplitude of the channel state information \cite{urllc_05}, \cite{MehrniaTWC}.

At the training phase and before the transmission starts, the channel samples are collected from receivers Rx$1$ and Rx$2$ and converted into $n$ independent and identically distributed (i.i.d.) sample sequences denoted by $X^{n}$ and $Y^{n}$, respectively, by applying declustering method \cite{MehrniaTWC}, \cite{evt_04}. The channel data in the training phase can be collected by using either the pilot signals or the records from the prior data transmissions. The recommended amount of required samples in the training phase is about $1/\epsilon$, where $\epsilon$ is the target error probability on the order of $10^{-9}$-$10^{-5}$ for the URC system.
Suppose the channel is stationary according to the Augmented Dickey-Fuller (ADF) test results. In that case, the channel tail distribution is modeled by applying EVT to the received signal powers and estimating the parameters of the UGPD fitted to the channel tail distribution. Otherwise, external factors leading to variation over time in the parameters of the UGPD are identified and utilized to partition the sequence into stationary groups of size $M$. Then, EVT is applied to each stationary sequence for estimating the shape and scale parameters of UGPD as a change-point function of time, as explained in detail in \cite{MehrniaTVT}. The parameters of UGPDs obtained for the received powers of Rx$1$ and Rx$2$ are then mutually used to determine the tail distribution of the joint probability distribution of multiple channel sequences $X^n$ and $Y^n$ by using MEVT techniques to characterize the statistics of the inter-relationships of extreme events \cite{Mehrnia_twcBGPD}. The transmitter assumes that the main source in the block fading channel is link outage \cite{urllc_02}, \cite{urllc_05}, \cite{urllc_noise}.

Upon deriving MEVT-based channel model, transmission strategies are determined to estimate the optimum transmission rate with the goal of fulfilling a predetermined $\epsilon$ reliability constraint such that

\begin{equation}
\label{eqn:pfRxn}
    p_F(R(X^n,Y^n)) \leq \epsilon,
\end{equation}
where $R(X^n,Y^n)$ denotes the transmission rate estimated based on the $n$ training received power samples of $X^n$ and $Y^n$, $F$ is the joint CDF of the training samples from channel data $X^n$ and $Y^n$, and $p_F(R(X^n,Y^n))$ is the outage probability at transmission rate $R(X^n,Y^n)$ defined as

\begin{equation}
\label{eqn:outageprob}
    p_{F}(R(X^n,Y^n)) = P\big[ R(X^n,Y^n)>\log_2(1+Z) \big],
\end{equation}
where $Z$ is any received power from test samples of receiver Rx$1$ or Rx$2$.
The transmitter determines the optimum rate as a function of $F$ as follows:
\begin{flalign}
\label{eqn:eoutage}
   R(F) &= sup\big\{ R(X^n,Y^n) \ge 0 : p_{F}(R(X^n,Y^n)) \le \epsilon\big\}\\
   &= \log_{2} \big(1+F^{-1}(\varepsilon_n)\big),
\end{flalign}
where $F^{-1}(\varepsilon_n)$ is the $\varepsilon_n$-quantile of $F$, and $\varepsilon_n$ is determined as a function of $\epsilon$ and the parameters of $F$ with the goal of satisfying constraint (\ref{eqn:pfRxn}).

Upon determining the optimal transmission rate using the proposed MEVT rate selection methodology at target error probability $\epsilon$, the confidence intervals of the estimated UGPD parameters are computed, considering some uncertainty in the parameter estimation by incorporating the probability of wrong decision $\alpha$. Subsequently, the statistics of transmission rate are estimated using MEVT, which incorporates the confidence intervals of the estimated UGPD parameters to satisfy the target error probability for a system equipped with spatial diversity operating in the ultra-reliable communication domain.

\section{MEVT-based Rate Selection Framework}
\label{sec:framework}
The goal of the MEVT-based rate selection framework is to determine an optimum transmission rate along with the confidence interval based on the estimation of the lower tail statistics of multiple channels for real-time communication in the ultra-reliable regime.
The proposed methodology involves a series of steps, including the conversion of the received power sample sequences into independent and identically distributed (i.i.d) samples. This is achieved through the application of the declustering method, which removes interdependency between the samples.
Then, the inter-relationship of bi-variate extremes is modeled by utilizing BGPD based on the MEVT-Poisson point process approach. Upon determining the optimum transmission rate based on the proposed MEVT-based rate selection framework, channel reliability is assessed based on the results of outage probability. Finally, the confidence interval is incorporated into the proposed rate selection framework by determining confidence intervals of the estimated parameters and the estimated optimum transmission rate for different values of training sample number $n$ to specify the minimum required training samples for optimal estimation of the transmission rate. 
The applied algorithm for the proposed rate selection framework incorporating CI is depicted in Fig.~\ref{fig:MEVT_rate_diagram} and explained in detail next.

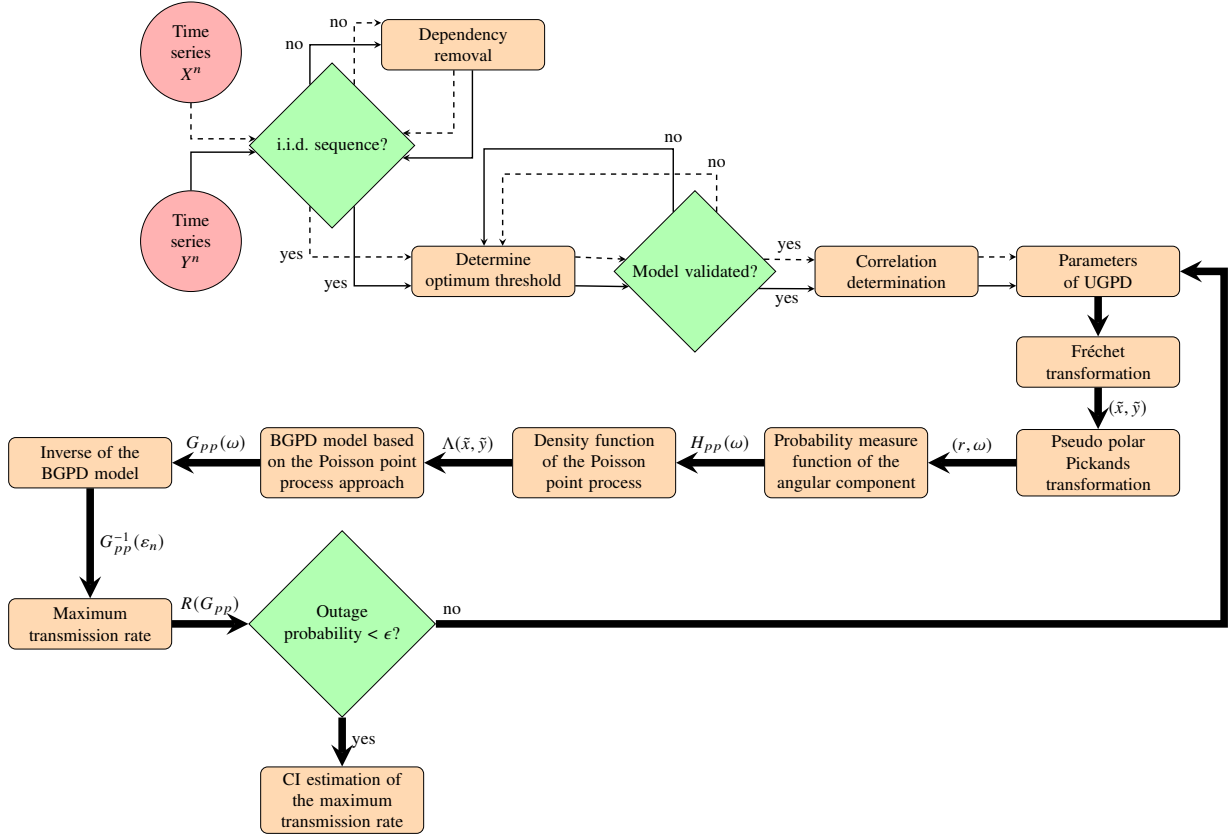
\begin{figure*}
    \centering

\scalebox{0.67}{
\begin{tikzpicture}[node distance=2.5cm]
\label{fig:maindiagram}

\node (start1) [startstop] {Time series $X^n$};
\node (start2) [startstop, below of=start1 , yshift=-1.25cm] {Time series $Y^n$};
\node (dec1) [decision, right of=start1, xshift=0.3cm, yshift=-1.85cm] {i.i.d. sequence?};
\node (pro1) [process, right of=dec1, yshift=2cm, xshift=0.1cm] {Dependency removal};
\node (pro3) [process, right of=dec1, yshift=-2.5cm, xshift=0.7cm] {Determine optimum threshold};
\node (dec3) [decision, right of=pro3, xshift=1.5cm] {Model validated?};
\node (pro4) [process, right of=dec3, xshift=1.5cm] {Correlation determination};
\node (pro5) [process, right of=pro4, xshift=1.5cm] {Parameters of UGPD};
\node (pro6) [process, below of=pro5, yshift=0.7cm] {Fr\'echet transformation};
\node (pro12) [process, below of=pro6, yshift=0.5cm] {Pseudo polar Pickands transformation};
\node (pro14) [process, left of=pro12, xshift=-2.5cm] {Probability measure function of the angular component};
\node (pro15) [process, left of=pro14, xshift=-2.5cm] {Density function of the Poisson point process};
\node (pro16) [process, left of=pro15, xshift=-2.5cm] {BGPD model based on the Poisson point process approach};
\node (pro17) [process, left of=pro16, xshift=-2.5cm] {Inverse of the BGPD model};
\node (pro18) [process, below of=pro17, yshift=-0.7cm] {Maximum transmission rate};
\node (dec4) [decision, right of=pro18, xshift=2.5cm] {Outage probability < $\epsilon$?};
\node (pro19) [process, below of=dec4, yshift=-1.0cm] {CI estimation of the maximum transmission rate};

\draw [arrow, dashed] (start1) |- (dec1.175);
\draw [arrow] (start2) |- (dec1.185);
\draw [arrow, dashed] (dec1.70) |- node[anchor=east] {no} (pro1.165);
\draw [arrow] (dec1.110) |- node[anchor=east] {no} (pro1);
\draw [arrow, dashed] (dec1.250) |- node[anchor=east] {yes} (pro3.170);
\draw [arrow] (dec1.290) |- node[anchor=east] {yes} (pro3.190);
\draw [arrow, dashed] (pro1.250) |- (dec1.10);
\draw [arrow] (pro1.290) |- (dec1.350);
\draw [arrow, dashed] (pro3.10) -- (dec3.170);
\draw [arrow] (pro3.350) -- (dec3.193);
\draw [arrow, dashed] (dec3.70) |- node[anchor=south] {no} ++(0,0.75) -|  (pro3.70);
\draw [arrow] (dec3.110) |- node[anchor=south] {no} ++(0,1.25) -|  (pro3.110);
\draw [arrow, dashed] (dec3.10) -- node[anchor=south] {yes} (pro4.172);
\draw [arrow] (dec3.345) -- node[anchor=north] {yes} (pro4.192);
\draw [arrow, dashed] (pro4.10) -- (pro5.170);
\draw [arrow] (pro4.350) -- (pro5.190);
\draw [arrow, line width=1.5mm] (pro5) -- (pro6);
\draw [arrow, line width=1.5mm] (pro6) -- node[anchor=west] {$(\tilde{x},\tilde{y})$} (pro12);
\draw [arrow, line width=1.5mm] (pro12) -- node[anchor=south] {($r,\omega$)} (pro14);
\draw [arrow, line width=1.5mm] (pro14) -- node[anchor=south] {$H_{pp}(\omega)$} (pro15);
\draw [arrow, line width=1.5mm] (pro15) -- node[anchor=south] {$\Lambda(\tilde{x},\tilde{y})$} (pro16);
\draw [arrow, line width=1.5mm] (pro16) -- node[anchor=south] {$G_{pp}(\omega)$} (pro17);
\draw [arrow, line width=1.5mm] (pro17) -- node[anchor=west] {$G^{-1}_{pp}(\varepsilon_n)$} (pro18);
\draw [arrow, line width=1.5mm] (pro18) -- node[anchor=south] {$R(G_{pp})$} (dec4);
\draw [arrow, line width=1.5mm] (dec4) -- node[anchor=west] {yes}  (pro19);
\draw [arrow, line width=1.5mm] (dec4) -- node[anchor=south] {no} ++(2.5,0) -- ++(15,0) |- (pro5);

\end{tikzpicture}
} 
    \caption{Flow diagram of the proposed MEVT-based rate selection framework.}
    \label{fig:MEVT_rate_diagram}
\end{figure*}

\subsection{Bi-variate Channel Estimation}
\label{sec:framework_channel}

At the training phase, the channel samples are collected from receivers Rx$1$ and Rx$2$ and converted into i.i.d. samples through the declustering approach. In this approach, the samples are folded into multiple clusters, and then the minimum of each cluster is extracted; the minima are considered i.i.d. samples \cite{MehrniaTWC}, \cite{evt_04}. 
Upon applying EVT to the resulting sequence of i.i.d. samples from individual channel sequences, the optimum thresholds are determined separately by applying two complementary methods based on EVT: mean residual life (MRL) and parameter stability methods \cite{MehrniaTWC}, \cite{evt_04}. 
According to the MRL method, the optimum threshold is the highest threshold below which the mean of values exceeding a given threshold is linear with respect to the threshold. On the other hand, the optimum threshold based on the parameter stability method is the highest threshold below which the estimated shape and modified scale parameters are linear against the threshold \cite{MehrniaTWC}, \cite{evt_04}.
The MRL method can be applied prior to the estimation of the tail distribution by using UGPD, while the parameter stability method should be applied after deriving the tail statistics by using UGPD.
Since the bi-variate analysis are limited to time periods where both thresholds are exceeded, in this study, only the exceedances that occur simultaneously are considered, while any exceedances that occur at different time intervals are excluded. Following this, the parameters of the UGPD are estimated for the exceedances of each channel data, using the maximum likelihood (ML) estimator with the corresponding optimal thresholds.
The UGPD fitted to the tail of sample sequence $X^n$ is formulated as

\begin{equation}
\label{eqn:gpd_dist}
    G_x(\tilde{\sigma}_x,\xi_x) = 1-\Big[1+\frac{\xi_x (u_x-x)}{\tilde{\sigma}_x}\Big]^{-1/\xi_x},
\end{equation}
where $(u_x-x)$ denotes the non-negative exceedance for observation $x \in X^n$ below optimum threshold $u$; and $\xi_x$ and $\tilde{\sigma}_x$ are the estimated shape and scale parameters of UGPD, respectively. 
The validity of the fitted UGPD model to the tail distributions is assessed by using the probability plots, including the probability-probability (PP) plot and the quantile-quantile (QQ) plot. The PP and QQ plots are graphical techniques that are employed to compare the empirical values with the corresponding modeled results. In the PP plot, the empirical CDF of the occurrence for each extreme value is plotted against the corresponding CDF obtained by the UGPD, while in the QQ plot, the empirical extreme quantile is plotted against the corresponding quantile obtained by the inverse of UGPD \cite{MehrniaTWC}. 
If the UGPD properly characterizes the extreme values exceeding threshold $u$, then both PP and QQ plots conform to the unit diagonal line, i.e., a straight line with the horizontal angle $45^\circ$. Thereafter, the inter-relationship of bi-variate extremes is modeled by applying the Poisson point process approach.

In bi-variate extreme analysis, the first step involves examining the presence of dependency within the tail samples and the total samples. The feasibility of spatial diversity is evaluated using the correlation between received powers in the total samples. Spatial diversity is considered reasonable if the correlation coefficient is within the range of $0.1$ and $0.5$ among the total samples, since otherwise fading occurs simultaneously at different links \cite{corcoef_01}. On the other hand, if the extremes of the two sample sequences are independent, i.e., the correlation coefficient approaches $0$, there is no need for bi-variate modeling. However, if there is correlation between the tail samples, it is necessary to investigate the inter-relationship of the extreme values of receivers.


In the Poisson point process approach, the channel data is transformed in two steps: i) Fr\'echet transformation and ii) Pickands coordinates.
The Fr\'echet transformation is applied on the obtained tail sequences of $X^n$ and $Y^n$ as follows:

\begin{equation}
\label{eqn:x_tilde}
    \tilde{x} = -\Big( \log \Big\{ 1-\zeta_{x} \Big[1+\frac{\xi_x (u_x-x)}{\tilde{\sigma}_{x}}\Big]^{-1/\xi_x} \Big\} \Big) ^{-1}, \hspace{0.3cm} x<u_x,
\end{equation}
and
\begin{equation}
\label{eqn:y_tilde} 
    \tilde{y} = -\Big( \log \Big\{ 1-\zeta_{y} \Big[1+\frac{\xi_y (u_y-y)}{\tilde{\sigma}_{y}}\Big]^{-1/\xi_y} \Big\} \Big) ^{-1}, \hspace{0.3cm} y<u_y,
\end{equation}
where $\zeta_{x} = Pr(x<u_x)$ and $\zeta_{y} = Pr(y<u_y)$ for optimum thresholds $u_x$ and $u_y$, respectively. Accordingly, the marginal distribution of variables $\tilde{x}$ and $\tilde{y}$ approximately have Fr\'echet distribution for $x < u_x$, $x \in X^n$, and $y< u_y$, $y \in Y^n$. After applying the Fr\'echet transformation, the pseudo polar Pickands transformation is applied by introducing pseudo polar radial $r$ and angular $\omega$ components as 

\begin{equation}
\begin{split}
    r = \frac{-\tilde{x}}{n} + \frac{-\tilde{y}}{n},\\
    \omega = \frac{-\tilde{x}/n}{-\tilde{x}/n-\tilde{y}/n},  
\end{split}
\end{equation}
where $n \in N$ is the length of vector $\tilde{X}^n$ and $\tilde{Y}^n$, and $\tilde{X}^n$ and $\tilde{Y}^n$ are the Fr\'echet transformation of $X^n$ and $Y^n$, respectively.
Upon determining the probability measure function of the Pickands angular coordinate $H_{pp}(\omega)$, satisfying the mean constraint $\int_{\omega=0}^1 \omega \, H_{pp}(d\omega) = 0.5$, the density function of the Poisson point process $\Lambda(\tilde{x},\tilde{y})$ is computed as

\begin{flalign}
\label{eqn:capitalLambdaThm}
\begin{split}
    \Lambda(\tilde{x},\tilde{y}) &= \int_{r=-\infty}^{0} \int_{\omega=0}^{1} 2 \frac{dr}{r^2} dH(\omega) \\ 
    &=  -2 \int_{\omega=0}^{1} \max \big(\frac{\omega}{-\tilde{x}/n},\frac{1-\omega}{-\tilde{y}/n}\big) \,H_{pp}(d\omega),
\end{split}
\end{flalign}
where $\Lambda(\tilde{x},\tilde{y})$ is defined on space $\{(0,\infty) \times (0,\infty) \backslash (0,\tilde{x}) \times (0,\tilde{y}) \}$, denoting space $\{(0,\infty) \times (0,\infty)\}$ excluding sub-space $\{(0,\tilde{x}) \times (0,\tilde{y})\}$.
Finally, the BGPD model based on the Poisson point process approach is determined as \cite{Mehrnia_twcBGPD}

\begin{equation}
\label{eqn:bgpd}
    G_{pp}(\tilde{x},\tilde{y}) = \exp(-\Lambda(\tilde{x},\tilde{y})).
\end{equation}

\subsection{Bi-variate Rate selection}
\label{sec:bi-rateselec}
The goal of this section is to choose the maximal transmission rate that achieves the target reliability $\epsilon$ by satisfying constraint (\ref{eqn:pfRxn}). 
The maximum rate is determined as a function of $G_{pp}$ as 

\begin{flalign}
\label{eqn:ReGppGeneral}
\begin{split}
   R(G_{pp}) &= sup\big\{ R(X^n+Y^n) \ge 0 : p_{F}(R(X^n+Y^n)) \le \epsilon\big\}\\
   \hfill \hspace{0.3cm}
   &= \log_{2} \big(1+G_{pp}^{-1}(\varepsilon_n)\big),
\end{split}
\end{flalign}
where $G_{pp}^{-1}(\varepsilon_n)$ is an estimate of a positive $\varepsilon_n$-quantile of the bi-variate channel $G_{pp}$. Here, the objectives are i) to determine $G_{pp}^{-1}(\varepsilon_n)$ and ii) to find $\varepsilon_{n}$ that maximizes the transmission rate $R(G_{pp})$ while satisfying the target reliability $\epsilon$ based on constraint (\ref{eqn:pfRxn}).

In order to determine $G_{pp}^{-1}(\varepsilon_n)$, we first obtain the inverse joint CDF function of the training samples from channel data $X^n$ and $Y^n$. 

\begin{theorem}
\label{theorem:rate}
    Let $G_{pp}(\hat{\tilde{x}},\hat{\tilde{y}}) \approx G_{pp}(\hat{\omega})$ be the estimated BGPD model fitted to the training samples $X^{n} = \{x_1,x_2,...,x_n\}$ and $Y^{n} = \{y_1,y_2,...,y_n\}$. Then, the inverse joint CDF function $F^{-1} \sim G_{pp}^{-1}(\varepsilon_n)$ can be obtained as 
    
    \begin{equation}
        G_{pp}^{-1}(\varepsilon_n) = H^{-1}_{\hat{\omega}} \Big\{ \frac{1}{2} max(-\hat{\tilde{x}}/n-\hat{\tilde{y}}/n) \ln{\varepsilon_n} \Big\},
    \end{equation}
    where $\hat{\Tilde{x}}$ and $\hat{\Tilde{y}}$ are the estimated Fr\'echet transformation of $G_x(\hat{\tilde{\sigma}}_x,\hat{\xi}_x)$ and $G_y(\hat{\tilde{\sigma}}_y,\hat{\xi}_y)$ fitted to $X^n$ and $Y^n$ training sequences, respectively, $\ln{(.)}$ function refers to the natural logarithm, $H^{-1}_{\hat{\omega}}(.)$ is the CDF-inverse of the $\beta$-distribution fitted to the pseudo-polar angular component $\hat{\omega}=\frac{\hat{\Tilde{x}}/n}{\hat{\Tilde{x}}/n+\hat{\Tilde{y}}/n}$.
\end{theorem}

\begin{proof}
    Assume that $G_{pp}(\hat{\tilde{x}},\hat{\tilde{y}})=exp(-\Lambda(\tilde{x},\tilde{y}))$ models the bi-variate GPD fitted to the training samples $X^{n}$ and $Y^{n}$. Then,
    
    \begin{flalign*}
    \begin{split}
        \Lambda(\tilde{x},\tilde{y}) &=  \int_{\hat{\omega}=0}^{1} \int_{\hat{r}=-\infty}^{0} 2 \frac{d\hat{r}}{\hat{r}^2} dH(\hat{\omega})\\
        &= \int_{\hat{\omega}=0}^{1} \int_{\frac{-\hat{\tilde{x}}/n}{min(-\hat{\tilde{x}}/n-\hat{\tilde{y}}/n)}}^{\frac{-\hat{\tilde{x}}/n}{max(-\hat{\tilde{x}}/n-\hat{\tilde{y}}/n)}} \frac{-2}{\hat{\tilde{x}}/n} d\hat{\omega} \, dH(\hat{\omega}).
    \end{split}
    \end{flalign*}
    In the above equation, $\frac{d\hat{r}}{\hat{r}^2}$ has been replaced by $\frac{-2}{\hat{\tilde{x}}/n} d\hat{\omega}$ since $\hat{\omega} = \frac{-\hat{\tilde{x}}/n}{-\hat{\tilde{x}}/n-\hat{\tilde{y}}/n} = \frac{-\hat{\tilde{x}}/n}{\hat{r}}$, and therefore $\hat{r}$ can be written as $\hat{r} = \frac{-\hat{\tilde{x}}/n}{\hat{\omega}}$. Additionally, the boundary of $\hat{r}$ has been computed within the range $[ \frac{-\hat{\tilde{x}}/n}{min(-\hat{\tilde{x}}/n-\hat{\tilde{y}}/n)} , \frac{-\hat{\tilde{x}}/n}{max(-\hat{\tilde{x}}/n-\hat{\tilde{y}}/n)} ]$ since
    
    \begin{equation*}
    \begin{split}
        min(-\hat{\tilde{x}}/n-\hat{\tilde{y}}/n) < \hat{r} < max(-\hat{\tilde{x}}/n-\hat{\tilde{y}}/n) \\ \Rightarrow   
        \frac{-\hat{\tilde{x}}/n}{min(-\hat{\tilde{x}}/n-\hat{\tilde{y}}/n)} < \frac{-\hat{\tilde{x}}/n}{\hat{r}} = \hat{\omega}  < \frac{-\hat{\tilde{x}}/n}{max(-\hat{\tilde{x}}/n-\hat{\tilde{y}}/n)}
    \end{split} 
    \end{equation*}
Moreover, $min(-\hat{\tilde{x}}/n-\hat{\tilde{y}}/n) \to -\infty$ and hence, $\Lambda(\tilde{x},\tilde{y})$ is further simplified as

\begin{flalign}
\label{eqn:lambdaH}
\begin{split}
    \Lambda(\tilde{x},\tilde{y}) &= -2 (\frac{1}{max(-\hat{\tilde{x}}/n-\hat{\tilde{y}}/n)}) \int_{\hat{\omega}=0}^{1} dH(\hat{\omega})\\
    &= \frac{-2}{max(-\hat{\tilde{x}}/n-\hat{\tilde{y}}/n)} H(\hat{\omega}),
\end{split}
\end{flalign}
On the other hand, since $G_{pp}(\hat{\tilde{x}},\hat{\tilde{y}})=exp(-\Lambda(A))$, $\Lambda(A) = -\ln{G_{pp}(\hat{\tilde{x}},\hat{\tilde{y}})}$, and therefore,  

\begin{align}
\label{eqn:Hw}
        -\ln{G_{pp}(\hat{\tilde{x}},\hat{\tilde{y}})} = \frac{-2}{max(-\hat{\tilde{x}}/n-\hat{\tilde{y}}/n)} H(\hat{\omega}), \\
        \Rightarrow \hat{\omega} = H^{-1}_{\hat{\omega}} \Big\{ \frac{max(-\hat{\tilde{x}}/n-\hat{\tilde{y}}/n)}{2} \ln{G_{pp}(\hat{\tilde{x}},\hat{\tilde{y}})} \Big\}
\end{align}
where $\hat{\omega}=\frac{-\hat{\tilde{x}}/n}{-\hat{\tilde{x}}/n-\hat{\tilde{y}}/n}$, the angular component, is a rational function defined only on $[0,1]$, and is fitted to the $\beta$-distribution \cite{pickands_02} with parameters $\hat{p}$ and $\hat{q}$, i.e., $\hat{\omega} \sim \beta(\hat{\omega};\hat{p},\hat{q})$, and hence, $H^{-1}_{\hat{\omega}}(.)$ is the CDF-inverse of the $\beta$-distribution fitted to the angular component $\hat{\omega}$.
\end{proof}

Since determining the inverse of $\beta$ function is a non-trivial task, the Newton–Raphson method is utilized to determine the inverse of $\beta(\hat{\omega};\hat{p},\hat{q})$ \cite{newton}. Newton–Raphson method is a recursive approximation technique for finding the root and inverse of a differentiable function.
 
The optimum transmission rate based on the proposed rate selection framework is then determined by substituting the results of Theorem~\ref{theorem:rate} into (\ref{eqn:ReGppGeneral}) as

\begin{equation}
\label{eqn:R_gpp}
    R(G_{PP}) = \log_2 \Big( 1+H^{-1}_{\hat{\omega}}\big\{ \frac{1}{2} max(-\hat{\Tilde{x}}/n-\hat{\Tilde{y}}/n) \ln{\varepsilon_n} \big\} \Big),
\end{equation}
for the maximum allowed error probability $\varepsilon_n$, $\varepsilon_n \leq \epsilon$, where $H^{-1}_{\hat{\omega}}(.)$ denotes the CDF-inverse of $\beta$-distribution fitted to the angular component $\hat{\omega}$ with shape parameters $\hat{p}$ and $\hat{q}$, denoted by $\beta(\hat{\omega};\hat{p},\hat{q})$.

\subsection{Outage Probability of the Transmission Rate}
\label{sec:outage_prob}
Suppose $R(G_{pp})$ is a valid transmission rate. In that case, the corresponding outage probability $p_{F}(R(G_{pp}))$ is expected not to exceed the target reliability $\epsilon$, according to constraint (\ref{eqn:pfRxn}).
The average outage probability for the transmission rate obtained in the previous section is obtained by substituting (\ref{eqn:R_gpp}) into the outage probability equation expressed in (\ref{eqn:outageprob}) as 

\begin{flalign}
\label{eqn:pfR}
\begin{split}
    p_F(R(G_{pp}))
    &= P \Big[ \log_2\big(1+H^{-1}_{\hat{\omega}}\big(\frac{1}{2} \, max(\hat{r}) \, \ln{\varepsilon_n} \big) \big) \\ 
    &> \log_2(1+Z)\Big]\\
    &= H_{Z}\big(H^{-1}_{\hat{\omega}}\big(\frac{1}{2} \, max(\hat{r}) \, \ln{\varepsilon_n}\big)\big),
\end{split}
\end{flalign}
where $\hat{r} = \frac{-\hat{\tilde{x}}}{n} + \frac{-\hat{\Tilde{y}}}{n}$, $Z = \frac{-\tilde{x}}{N} + \frac{-\Tilde{y}}{N}$, $N$ is the total number of samples including training and test samples, and $H_{Z}(.)$ is the CDF of $\beta$-distribution with parameters $p$ and $q$ fitted to $Z$, i.e., $\beta(Z;p,q)$, assuming the channel is perfectly known. In order to determine the maximum allowed error probability $\varepsilon_n$, the constraint (\ref{eqn:pfRxn}) should be satisfied as 

\begin{equation}
\label{eqn:varepsi}
    \varepsilon_n = exp \Big\{  \frac{2}{max(\hat{r})} \, H_{H^{-1}_{\epsilon} (p,q)} \big( \hat{p},\hat{q}  \big)  \Big\},
\end{equation}
where $H^{-1}_{\epsilon} (p,q)$ denotes the $\epsilon$-quantile of the $\beta$-distribution with parameters $p$ and $q$ fitted to $Z$, and $H_{H^{-1}_{\epsilon} (p,q)} \big( \hat{p},\hat{q}\big)$ is CDF of $\beta$-distribution with parameters $\hat{p}$ and $\hat{q}$ fitted to $\hat{\omega}$ and computed at point $H^{-1}_{\epsilon} (p,q)$.

\subsection{Confidence interval}
The confidence interval for the transmission rate can be calculated by deriving confidence intervals for the estimated UGPD parameters. Although various methods can be used for computing confidence intervals, bootstrapping methods are mostly preferred due to their higher accuracy and avoidance of assumptions on the properties of the distribution of the original samples. Bootstrapping methods generate confidence intervals for statistical parameters by drawing pseudo-samples either from the original sample sequence or a model fitted to the original sample.
In the realm of statistics, conventional bootstrapping-based methods are highly dependent on the asymptotic normality of parameter estimate to construct a standard confidence interval for estimation of the parameter $\hat{\theta}$.
However, a standard confidence interval might not be a proper way of considering the estimation uncertainty for many applications, especially when the distribution of the parameter estimation, i.e., $\hat{\theta}$, is not normal. Alternatively, the bootstrap-based bias-corrected and accelerated (BCA) method has been demonstrated to provide higher accuracy with narrower CIs correcting (i) the non-normality of $\hat{\theta}$ by utilizing the bootstrap distribution where a parametric model assumption is made for $\hat{\theta}$, (ii) the bias of $\hat{\theta}$ by using the bias correction parameter, and (iii) the non-constant standard error of $\hat{\theta}$ by introducing the acceleration parameter \cite{ci_10}, \cite{bootci_02}. 

\subsubsection{Fundamentals of bootstrap-based BCA}
\label{sec:bca-ci}
$B$ bootstrap samples are generated at each bootstrapping round with replacement, and then an estimate of the parameter is obtained as $\hat{\theta}_b$ at the $b$-th round for $b \in \{1,2,...,B\}$. 
The BCA method is based on the assumption that there is a monotone transformation $\hat{\phi} = h(\hat{\theta})$ such that $\hat{\phi} \sim N(\phi - z_0^{\hat{\theta}} \tau_{\phi}, \tau_{\phi}^2)$, where the transformed variable $\hat{\phi}$ follows a normal distribution with mean $\phi - z_0^{\hat{\theta}} \tau_{\phi}$ and variance $\tau_{\phi}^2$, $\phi$ is the median of the distribution of the bootstrap estimates, $\tau_{\phi} = 1+a^{\hat{\theta}}\phi$, and bias correction $z_0^{\hat{\theta}}$ and acceleration $a^{\hat{\theta}}$ factors are constant and estimated based on the CDF of normal distribution and jackknife re-sampling, respectively.

The bias correction factor $z_0^{\hat{\theta}}$ is estimated as 

\begin{equation}
\label{eqn:biascorec}
    z_0^{\hat{\theta}} = \Phi^{-1}(\hat{F}_e(\hat{\theta})),
\end{equation}
where $\Phi$ is the standard normal CDF and  $\hat{F}_e(.)$ is the empirical probability of the estimated parameter $\hat{\theta}$ and calculated as 

\begin{equation}
    \hat{F}_e = P(\hat{\theta}_b<\hat{\theta}) = \frac{m}{B}, \hspace{0.3cm} b=1,...,B,
\end{equation}
in which $m$ is the number of estimated parameters of bootstrapping $\hat{\theta}_{b}$ smaller than the estimated parameter based on the original samples $\hat{\theta}$. Please note that $\hat{\theta}$ is the estimated parameter from the original samples while $\hat{\theta}_{b}$ is the estimated parameter from bootstrapping the original samples. 

The acceleration factor $a^{\hat{\theta}}$ can be obtained based on jackknife re-sampling (or leave-one-out procedure) \cite{bootci_02}\nocite{ci_01}-\cite{ci_02} by incorporating $n$ replicates of the original samples, where $n$ is the total number of the original samples. Jackknife resampling involves generating a series of jackknife replicates by repeatedly leaving out one observation from the original sample. Specifically, the first replicate is generated by excluding the first observation, the second replicate is generated by excluding the second observation, and so on until $n$ replicates of size $n-1$ are obtained. For each sample number $j$, $j \in \{1,2,...,n \}$, the estimated parameter $\hat{\theta}^{j}$ is obtained. Then, the average of these estimations is calculated as $\bar{\theta} = \frac{1}{n}\sum_{j=1}^{n}\hat{\theta}^{j}$.
The acceleration factor is computed as 

\begin{equation}
\label{eqn:accefact}
    a^{\hat{\theta}} = \frac{\sum_{j=1}^{n}[\hat{\theta}^{j}-\bar{\theta}]^{3}}{6\{ 
\sum_{j=1}^{n} [\hat{\theta}^{j}-\bar{\theta}]^{2} \}^{\frac{3}{2}}}.
\end{equation}

Finally, the confidence interval is constructed as 

\begin{equation}
\label{eqn:thetahat}
    \hat{\theta} = [\theta_{L},\theta_{U}],
\end{equation}
where the estimated $\hat{\theta}_{b}$'s are ordered from smallest to the largest, $L = a_1^{\hat{\theta}} \times (B+1)$, and $U = a_2^{\hat{\theta}} \times (B+1)$ for $a_1^{\hat{\theta}}$ and $a_2^{\hat{\theta}}$ defined as 

\begin{equation}
\label{eqn:a1_coef}
    a_1^{\hat{\theta}} = \Phi \bigg\{  z_0^{\hat{\theta}} + \frac{z_0^{\hat{\theta}} - z^{\hat{\theta}}_{1-\frac{\alpha}{2}}}{1-a^{\hat{\theta}}(z_0^{\hat{\theta}} - z^{\hat{\theta}}_{1-\frac{\alpha}{2}})}  \bigg\},
\end{equation}
and 

\begin{equation}
\label{eqn:a2_coef}
    a_2^{\hat{\theta}} = \Phi \bigg\{  z_0^{\hat{\theta}} + \frac{z_0^{\hat{\theta}} + z^{\hat{\theta}}_{1-\frac{\alpha}{2}}}{1-a^{\hat{\theta}}(z_0^{\hat{\theta}} + z^{\hat{\theta}}_{1-\frac{\alpha}{2}})}  \bigg\},
\end{equation}
where $z^{\hat{\theta}}_{1-\frac{\alpha}{2}}$ is the $100(1-\frac{\alpha}{2})$ percentile of the standard normal distribution for significance level (or probability of wrong decision) $\alpha$ \cite{ci_02}.

\subsubsection{BCA in CI determination of UGPD parameters}
\label{sec:bca_ugpdci}
The estimated scale and shape parameters of UGPD are not exactly $\hat{\tilde{\sigma}}$, $\hat{\xi}$ and typically take values within the confidence intervals $[\hat{\tilde{\sigma}}_L,\hat{\tilde{\sigma}}_U]$ and $[\hat{\xi}_L,\hat{\xi}_U]$, respectively. The BCA-based algorithm that is utilized to determine the upper and lower bounds of the CI of the shape and scale parameters are given in Algorithm~\ref{alg:bca} and is described in detail as follows. The inputs of the algorithm are the observation sample set {$X^n$}; the significance level $\alpha$, depending on the application and how much error is allowed in the parameter estimation; and the number of newly generated data sets in the bootstrapping process denoted by $B$. The algorithm starts by bootstrapping the original data set for $B$ times and storing the new data sets in the vectors denoted by {$X^n_b$}, $b \in \{1,2,...B\}$ (Lines~$1-2$). The UGPD is fitted to the tail of data set {$X^n_b$}, and the parameters of UGPD are estimated as $\hat{\tilde{\sigma}}_b$ and $\hat{\xi_b}$ for $b \in \{1,2,...B\}$ (Line~$3$). Upon estimating the UGPD parameters for each data set {$X^n_b$}, the bias correction factor is obtained based on (\ref{eqn:biascorec}) for the estimated scale and shape parameters of UGPD as $z_0^{\hat{\tilde{\sigma}}}$ and $z_0^{\hat{\xi}}$, respectively (Line~$5$). The corresponding acceleration factors of the scale and shape parameters (Line~$10$) based on (\ref{eqn:accefact}) are determined by first applying the jackknife method and obtaining $n-1$ replica of {$X^n$} by leaving out the first $j$ samples of {$X^n$}, denoted by {$X^{n-j}$} (Line~$6-7$) and then, estimating the scale and shape parameters of fitted UGPD to each data set {$X^{n-j}$} as $\hat{\tilde{\sigma}}^j$ and $\hat{\xi}^j$, respectively (Line~$8$). Next, the $\{(1-\alpha) \times 100\} \%$ confidence intervals of the scale and shape parameters, denoted by $[\tilde{\sigma}_L,\tilde{\sigma}_U]$ and $[\xi_L$,$\xi_U]$, respectively, are estimated for the significance level $\alpha$ based on (\ref{eqn:thetahat}) (Line~$12$), through the determination of coefficients $a_1^{\hat{\tilde{\sigma}}}$, $a_2^{\hat{\tilde{\sigma}}}$, $a_1^{\hat{\xi}}$, and $a_2^{\hat{\xi}}$ according to (\ref{eqn:a1_coef}) and (\ref{eqn:a2_coef}) for the scale and shape parameters (Line~$11$). The same procedure is applied to the sample set $Y^n$ to compute the corresponding CI of the estimated parameters of UGPD fitted to $Y^n$.

\begin{algorithm}[ht]
\caption{\textbf{BCA-based CI Determination Algorithm}}
\label{alg:bca}
\begin{algorithmic}
\STATE \textbf{Input}: {$X^n$} $= [x_{1}, x_{2},...,x_{n}]$, $\alpha$, and $B$;
\STATE \textbf{Output}: $\tilde{\sigma}_L$, $\tilde{\sigma}_U$, $\xi_L$, and $\xi_U$;
\end{algorithmic}
\begin{algorithmic}[1]
\FOR{b = 1:B}
 \STATE obtain bootstrap samples with replacement of {$X^n$}, as
 {$X^n_b$};
 \STATE estimate $\hat{\tilde{\sigma}}_b$ and $\hat{\xi}_b$, by fitting the UGPD to the tail samples of {$X^n_b$};
\ENDFOR
\STATE determine the bias correction factors $z_0^{\hat{\tilde{\sigma}}}$ and $z_0^{\hat{\xi}}$
\FOR{$j = 1:n-1$}
\STATE leave out the first $j$ samples of {$X^n$}, as {$X^{n-j}$};
\STATE estimate $\hat{\tilde{\sigma}}^j$ and $\hat{\xi}^j$, by fitting UGPD to {$X^{n-j}$};
\ENDFOR
\STATE determine acceleration factors $a^{\hat{\tilde{\sigma}}}$ and $a^{\hat{\xi}}$;
\STATE determine coefficients $a_1^{\hat{\tilde{\sigma}}}$, $a_2^{\hat{\tilde{\sigma}}}$, $a_1^{\hat{\xi}}$, and $a_2^{\hat{\xi}}$;
\STATE determine the CIs of the $(1-\alpha)$ BCA-CI, $[\tilde{\sigma}_L,\tilde{\sigma}_U]$ and $[\xi_L$,$\xi_U]$, for the UGPD parameters.
\RETURN $\tilde{\sigma}_L$, $\tilde{\sigma}_U$, $\xi_L$, and $\xi_U$;
\end{algorithmic}
\end{algorithm}

\subsubsection{BCA in CI determination of the transmission rate}
\label{sec:bca_rateci}
Upon determining the upper and lower bounds of the estimated CI corresponding to the scale and shape parameters of $X^n$ and $Y^n$, the $\Tilde{x}$, $\Tilde{y}$, $R(G_{pp})$, are estimated within the intervals $[\tilde{x}_L,\tilde{x}_U]$, $[\tilde{y}_L,\tilde{y}_U]$, $[R_{L}, R_{U}]$, respectively. The estimated transmission rate, $R(G_{pp})$, is a function of the pseudo-polar angular component $\hat{\omega}$, or equivalently Fr\'echet transformations $\hat{\Tilde{x}}$ and $\hat{\Tilde{y}}$. The CI of the angular component and the Fr\'echet transformed variables are functions of the confidence intervals of the scale and shape parameters.
Therefore, the transmission rate $R(G_{pp})$ is estimated to be within the interval $[R_{L}, R_{U}]$, considering the confidence intervals of the estimated Pareto parameters.

Referring to (\ref{eqn:x_tilde}) and (\ref{eqn:y_tilde}), the upper (lower) bound of the scale parameters $\hat{\sigma}_x$ and $\hat{\sigma}_y$ along with the lower (upper) bound of the shape parameter $\hat{\xi}_x$ and $\hat{\xi}_y$ provide the lower (upper) bound of the Fr\'echet transformation parameters $\hat{\Tilde{x}}$ and $\hat{\Tilde{y}}$, respectively. Since the transmission rate obtained in Section~\ref{sec:bi-rateselec} is a function of the CDF of angular component $\hat{\omega}$, and 

\begin{equation*}
    \hat{\omega} = \frac{-\hat{\Tilde{x}}/n}{-\hat{\Tilde{x}}/n-\hat{\Tilde{y}}/n},
\end{equation*}
the upper (lower) bound of the $\hat{\Tilde{y}}$ in conjunction with the lower (upper) bound of $\hat{\Tilde{x}}$, results in the lower (upper) bound of $\hat{\omega}$. Moreover, since the CDF and CDF-inverse of $\hat{\omega}$ are monotonically increasing functions, the lower (upper) bound of $\hat{\omega}$ affects the lower (upper) bound of the estimated transmission rate $R(G_{pp})$ as follows:

\begin{equation}
\label{eqn:R_gpp_CI}
    R(G_{PP})_{L/U} = \log_2 \Big( 1+H^{-1}_{\hat{\omega}_{L/U}}\big\{ \frac{1}{2} max(-\hat{\Tilde{x}}/n-\hat{\Tilde{y}}/n) \ln{\varepsilon_n} \big\} \Big),
\end{equation}
where $R(G_{PP})_{L/U}$ and $H^{-1}_{\hat{\omega}_{L/U}}$ denotes the lower(L)/upper(U) bound of the estimated transmission rate $R(G_{PP})$ and CDF-inverse of the estimated Pickands angular component $H^{-1}_{\hat{\omega}}$, respectively.

It is worth noting that although $R(G_{pp})$ is a function of $max(-\hat{\Tilde{x}}/n-\hat{\Tilde{y}}/n)$ and $\varepsilon_n$, these parameters do not affect the CI of the transmission rate for the following reasons: For the first parameter, the $max$ function included in the $\varepsilon_n$ formula in (\ref{eqn:varepsi}) cancels the effect of $max(-\hat{\Tilde{x}}/n-\hat{\Tilde{y}}/n)$; for the second parameter, $\varepsilon_n$ approaches a fixed value equal to the target error probability $\epsilon$ for large sample numbers.

The complexity of the proposed MEVT-based rate selection framework is $O(n \, N_{Tx} \, N_{Rx})$, where $N_{Tx}$ and $N_{Rx}$ are the number of transmitters and receivers, respectively, and $n$ is the number of the training samples for individual channel sequences.

\section{Numerical Results}
\label{sec:numerical_results}
This section aims to evaluate the proposed rate selection framework compared to the traditional extrapolation-based approaches for a MIMO-URLLC system. In the following, we first provide the benchmark algorithm based on the traditional extrapolation-based approaches and simulation platform in Sections~\ref{sec:numres_benchmark} and \ref{sec:numres_simplat}, respectively. Then, in Section~\ref{sec:performance_rate}, we evaluate the performance of the proposed MEVT-based rate selection framework in determining the maximum transmission rate for URLLC in spatial diversity, evaluating the system reliability by using the outage probability metric and comparing it to the traditional extrapolation-based method. Finally, in Section~\ref{sec:performance_ci}, the estimation of the rate confidence interval is obtained by computing the confidence interval for the parameters of BGPD fitted to the tail distribution of the channel data, iterating over multiple sample sizes and values of the significance level $\alpha$.

\subsection{Benchmark Extrapolation-based Algorithm}
\label{sec:numres_benchmark}

In the conventional extrapolation-based method, the distribution of the current channel data is estimated for the reliability order ranging from $10^{-3}$ to $10^{0}$ PER, based on individual channel data sequences \cite{vehicular_01}-\cite{vehicular_02}. Subsequently, the tail distributions of these sequences are estimated by extrapolating the obtained results towards the ultra-reliable region of $10^{-9}$ to $10^{-5}$ PER \cite{urllc_02}.
After fitting different distributions to the samples in each observation sequence, the optimal distribution identified in our study is the Gaussian distribution, based on the Akaike information criterion/Bayesian information criterion (AIC/BIC) metric \cite{Mehrnia_twcBGPD}. In order to estimate the tail distribution of joint probability and establish a fair comparison between our proposed approach with the extrapolation-based method, we apply the Fr\'echet transformation, and then the Pseudo-polar Pickands transformation to the extrapolated results \cite{extrpltd_01}. Accordingly, the probability measure function of the corresponding angular component of Pickands transformation is modeled using the Beta distribution. Finally, the optimum transmission rate associated with the targeted error probability $\epsilon$ is modeled as

\begin{equation}
\label{eqn:R_extrpltd}
    R(F_{ep}) = \log_2 \Big( 1+H_{ep}^{-1}\big\{ \frac{1}{2} max(-\hat{\Tilde{x}}_{ep}/n-\hat{\Tilde{y}}_{ep}/n) \ln{\varepsilon^{\dagger}_n} \big\} \Big),
\end{equation}
where $F_{ep}$ is the extrapolated CDF, $\hat{\Tilde{x}}_{ep}$,$\hat{\Tilde{y}}_{ep}$ are the corresponding Fr\'echet variables of extrapolated results, $H_{ep}$ is the probability measure function of the angular component of extrapolated results, and $\varepsilon^{\dagger}_n$ is the maximum allowed error probability for the extrapolation-based method calculated as

\begin{equation}
\label{eqn:vare_extrpltd}
    \varepsilon^{\dagger}_n = exp \Big\{  \frac{2}{max(\hat{r}_{ep})} \, H^{ep}_{H^{-1}_{\epsilon ,ep} (p_{ep},q_{ep})} \big( \hat{p}_{ep},\hat{q}_{ep}  \big)  \Big\},
\end{equation}
where $\hat{r}_{ep} = \hat{\tilde{x}}_{ep} + \hat{\tilde{y}}_{ep} $, $H^{-1}_{\epsilon ,ep} (p_{ep},q_{ep})$ denotes the $\epsilon$-quantile of the $\beta$ distribution with parameters $p_{ep}$ and $q_{ep}$ fitted to the probability measure function of the angular component for the extrapolated results, and $H^{ep}_{H^{-1}_{\epsilon ,ep} (p_{ep},q_{ep})} \big( \hat{p}_{ep},\hat{q}_{ep}\big)$ is CDF of $\beta$ distribution with parameters $\hat{p}_{ep}$ and $\hat{q}_{ep}$ computed at point $H^{-1}_{\epsilon , ep} (p_{ep},q_{ep})$.

\subsection{Measurement Setup and Simulation Platform}
\label{sec:numres_simplat}
The proposed methodology is implemented on the dataset obtained from the engine compartment of a Fiat Linea vehicle, as depicted in Fig.~\ref{fig:engin}, utilizing a Vector Network Analyzer (VNA) (R$\And$S$\textsuperscript{\textregistered}$ ZVA$67$), as shown in Fig.~\ref{fig:vna}. The VNA is connected to the transmitter and receivers, Rx$1$ and Rx$2$, via the R$\And$S$\textsuperscript{\textregistered}$ ZV-Z$196$ and PE$361$ port cables, respectively, which have a length of $610$~mm. The transmitter is an omnidirectional antenna operating in the frequency range of $58$~GHz to $63$~GHz, with a nominal gain of $0$~dBi. The receivers are horn antennas operating in the frequency range of $50$-$75$ GHz, with a nominal gain of $24$~dBi, and horizontal and vertical half power beamwidths of $11^\circ$ and $9.5^\circ$, respectively.
The locations of the transmitter and receiver antennas are selected out of the possible locations for the wireless sensors located within the engine compartment, namely locations $13$ and $2\&4$ in \cite{demir2013engine}-Fig.~$1$, respectively, and shown in Fig.~\ref{fig:engin}.
The antennas are interconnected with the coaxial cables via a waveguide that operates within the frequency range of $50$-$65$~GHz. The waveguide exhibits an insertion loss of $0.5$dB and impedance of $50$$\Omega$.
It is worth noting that an omni-directional antenna at the
transmitter along with the directional antenna at two receivers are utilized to consider the usage of diversities in an ultra-reliable communication system and study the inter-relationships of extreme events for a MIMO system under an intra-vehicular communication. \textsc{MATLAB} is used to implement the proposed framework on the collected data.

\begin{figure}[h]
\centering
\captionsetup[subfigure]{labelformat=empty}
     \begin{center}
        \subfloat[(a)]{%
            \label{fig:engin}
            \includegraphics[width=0.7\columnwidth]{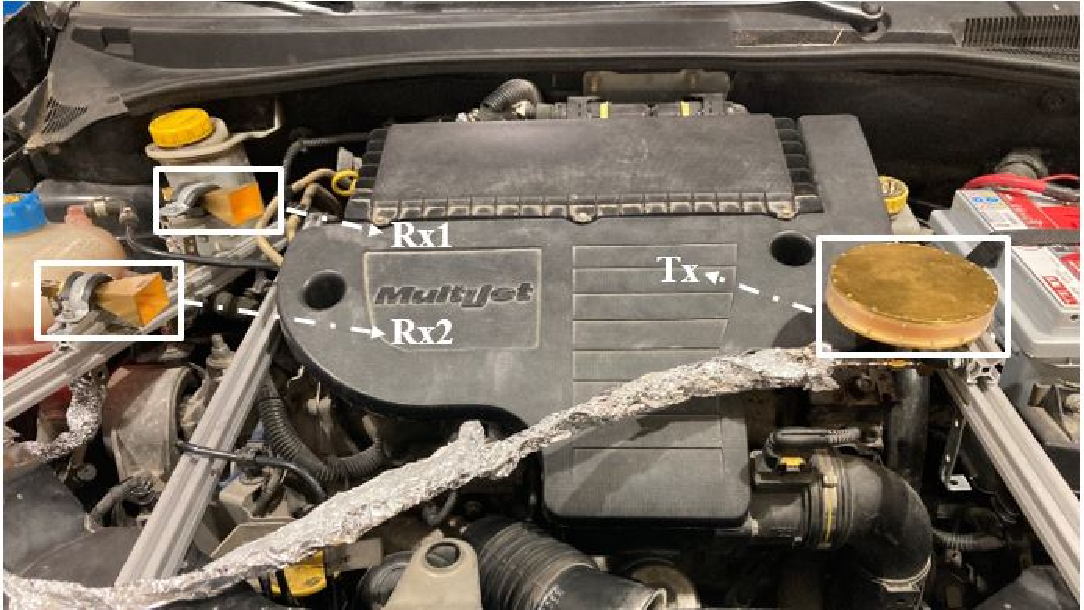}
        }\\%
        \subfloat[(b)]{%
            \label{fig:vna}
            \includegraphics[width=0.7\columnwidth]{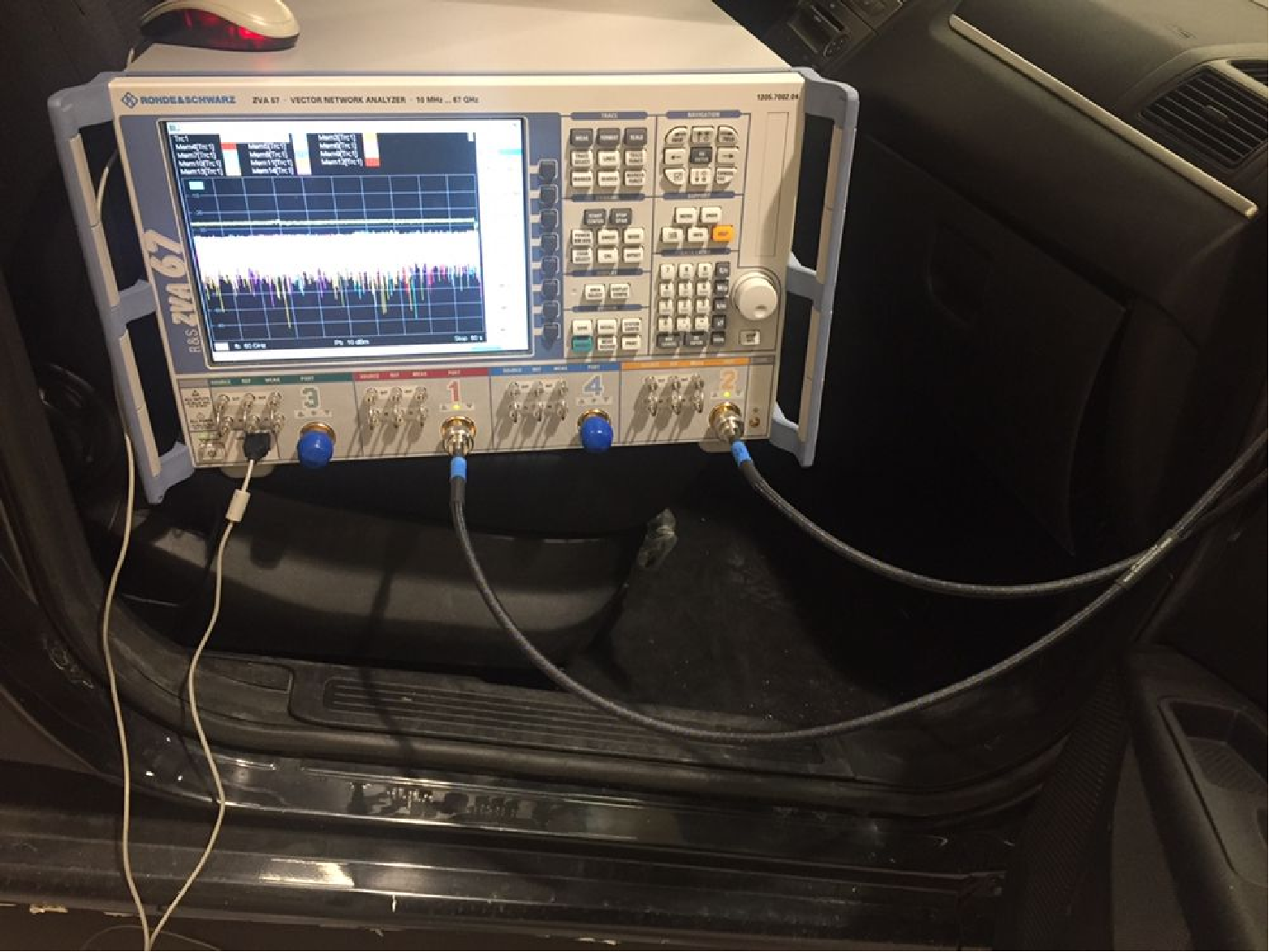}
        }\\ 
    \end{center}
    \caption{Measurement setup with the transmitter (TX) and receiver (RX) antennas located in the engine compartment of Fiat Linea: (a) Engine compartment, and (b) VNA setup.}
   \label{fig:probplotsArima}
\end{figure}

We categorize the samples from receivers Rx$1$ and Rx$2$ into two stationary groups: The first group, $Gr_1$, corresponds to driving on the smooth road, and the second group, $Gr_2$, corresponds to driving on the ramp. The optimum thresholds for the received power samples from receivers Rx$1$ and Rx$2$ are determined for each group of stationary data sequences based on the optimum threshold determination process in Section~\ref{sec:framework_channel}. Accordingly, for the stationary group $1$, the optimum thresholds of Rx$1$ and Rx$2$ are $-7$~dBm and $-6.8$~dBm, respectively, and for the stationary group $2$, the optimum thresholds of Rx$1$ and Rx$2$ are $-8$~dBm and $-27$~dBm, respectively. In the following, we only present the results of $Gr_2$ due to the limited number of pages. However, a similar observation has been made in the sample data of $Gr_1$.

\subsection{Transmission Rate}
\label{sec:performance_rate}
Fig.~\ref{fig:ratesample} shows the transmission rate of BGPD fitted to the filtered i.i.d. samples of the group $1$ at different sample numbers, and error probabilities $\epsilon \in \{10^{-3},10^{-4},10^{-5}\}$ compared to the extrapolation-based results. 
To address the same target error probability $\epsilon$, a higher value of optimum transmission rate is obtained based on the proposed MEVT-based framework, up to $10^{3}$ more than the extrapolation-based benchmark algorithm. The gap between their rates increases as the error probability increases. Additionally, according to the results obtained based on the extrapolation method at $\epsilon = 10^{-3}$, the transmission rate drops significantly after some sample number since $\varepsilon_{n}^{\dagger}$ approaches $0$, meaning that the maximum allowed error probability $\varepsilon_{n}^{\dagger}$ is estimated $0$ instead of $\sim 10^{-3}$ which is the target error probability. Besides, the estimated transmission rate is not a function of the sample number for a large number of samples and varies only as a function of the target error probability $\epsilon$. These results are consistent with the results of the EVT-based rate selection framework for single-input single-output (SISO)-URC presented in \cite{MehrniaTVTRate} and the extrapolation-based rate selection framework in \cite{urllc_05}.

The outage probability of the estimated transmission rate for $R(G_{pp})$ at different target error probability $\epsilon$ has also been computed. The outage probability results are the same for the proposed framework and the traditional extrapolation-based method. However, to experience the same outage probability, the traditional method requires a lower transmission rate which is neither power nor time-efficient. Therefore, by applying the proposed MEVT-based rate selection framework, we achieve a higher transmission rate while addressing the ultra-reliable constraint in which the outage probability does not exceed the target error probability.

\begin{figure}[ht]
\centering
\includegraphics[width=0.99\columnwidth]{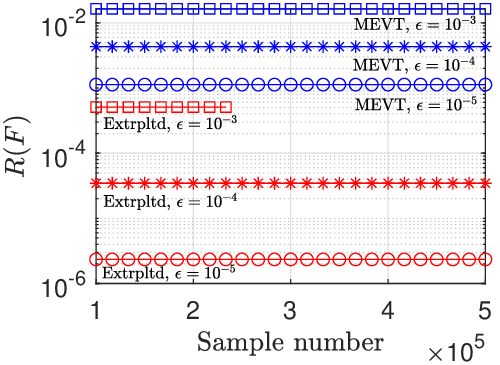}

    \caption{The transmission rate of BGPD fitted to the filtered i.i.d. samples of the group $1$ at different sample numbers, and targeted PER $\epsilon \in \{10^{-5},10^{-4},10^{-3}\}$.}
   \label{fig:ratesample}
\end{figure}

\subsection{Confidence Interval}
\label{sec:performance_ci}
Fig.~\ref{fig:sigmaxici} shows the estimated UGPD parameters along with the lower and upper bounds of CIs for $\alpha= \{0.05,0.2,0.5\}$ and different sample numbers in both receivers Rx$1$ and Rx$2$ for stationary group $1$. Compared to the standard CI results presented in \cite{MehrniaEucnc}, the confidence interval obtained by the non-bootstrap method is wider in general above a certain number of samples, and the estimated GPD parameters are almost constant. Moreover, it is observed that the MLEs of the scale ($\sigma$) and shape ($\xi$) parameters are subject to significant uncertainty when the sample size is limited. Nonetheless, the level of uncertainty reduces progressively with an increase in the sample size. This is consistent with the findings from the standard CI results as reported in \cite{MehrniaEucnc}. This underscores the trade-off between the uncertainty of the GPD parameters and the expenses involved in gathering data. In contrast to the conventional non-parametric technique for rate selection, which mandates noticeably more training on the channel, as outlined in \cite{urllc_05}, our framework based on MEVT appropriately infers the rate with a certain level of assurance, i.e., CI, using a smaller sample size of around $1/\epsilon$.

\begin{figure}[h]
\centering
\captionsetup[subfigure]{labelformat=empty}
     \begin{center}
        \subfloat[(a)]{%
            \label{fig:sigma1ci}
            \includegraphics[width=0.9\columnwidth]{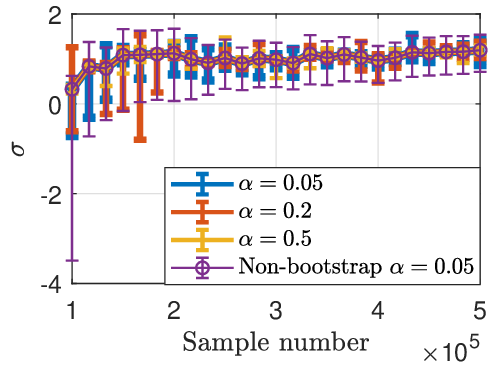}
        }\\%
        \subfloat[(b)]{%
            \label{fig:sigma2ci}
            \includegraphics[width=0.9\columnwidth]{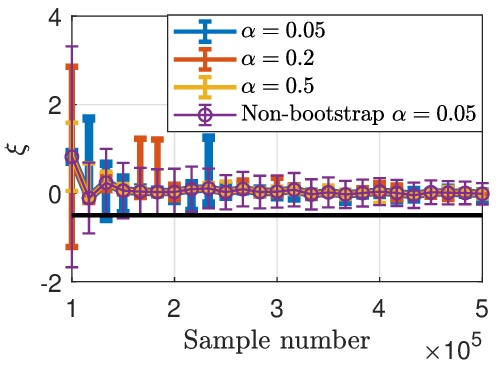}
        }\\
    \end{center}
    \caption{The estimated Pareto parameters along with their CI considering $\alpha=0.05,0.2,0.5$ for stationary group $1$, receivers Rx$1$ and Rx$2$, at different sample numbers: (a) Scale parameter and the corresponding CI for Rx$1$, and (b) Shape parameter and the corresponding CI for Rx$1$; the horizontal black line refer to the $-0.5$ minimum acceptable value for the shape parameter of UGPD.}
   \label{fig:sigmaxici}
\end{figure}

Fig.~\ref{fig:beta} illustrates the empirical CDF of the angular component of Pickands transformation along with the $0.01$ confidence interval (the dotted plots), compared to the estimated Beta distribution fitted to the probability measure function of the Pickands' angular coordinate for two different sample numbers. $\beta_1$ and 
$\beta_{2}$ correspond to the fitted $\beta$-distribution for the sample numbers $n_1 \approx 1.3 \times 10^{5}$ and $n_2 \approx 5 \times 10^{5}$, respectively, where the estimation of UGPD parameters and the corresponding angular component $\omega$ become stable. It is observed that both $\beta_1$ and $\beta_2$ distributions fit the empirical results within the $\alpha = 0.01$ confidence interval. Consequently, utilizing MEVT for modeling the multivariate channel and determining the associated transmission rate can significantly reduce the required number of samples by a minimum of $10$-fold compared to traditional extrapolation and data-driven models.

\begin{figure}[ht]
\centering
\includegraphics[width=0.9\columnwidth]{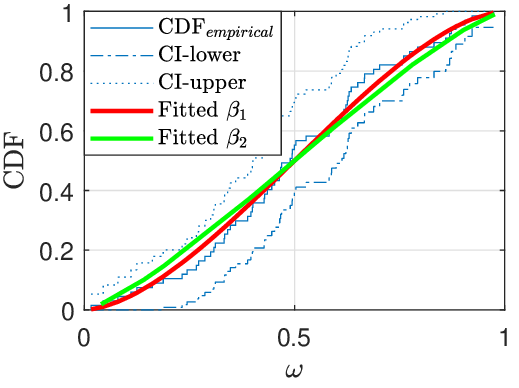}

    \caption{The CDF of pseudo-polar Pickands angular component based on empirical results along with $0.01$ confidence interval, the Beta distribution $\beta_1$ fitted to the sample numbers $\approx 1.3 \times 10^{-5}$, and the Beta distribution $\beta_2$ fitted to the sample number $\approx 5 \times 10^{-5}$.}
   \label{fig:beta}
\end{figure}

Fig.~\ref{fig:ratesampleci} depicts the estimated transmission rate based on the proposed MEVT-base rate selection framework for different sample numbers with the corresponding CI for $\alpha = \{0.05,0.2,0.5\}$ and target error probabilities $\epsilon \in \{10^{-5},10^{-4},10^{-3}\}$. 
The estimated transmission rate is constant with respect to the sample size, the same as what we have observed in Fig.~\ref{fig:ratesample}. The confidence interval is also constant with respect to the sample size except for low sample numbers, where the estimation of UGPD parameters suffers from uncertainty and varies within a wider confidence interval. However, above a certain sample size, the estimation of CIs becomes stable and converges. These observations are in good agreement with what we observed for the SISO system illustrated in \cite{MehrniaEucnc}.

\begin{figure}[ht]
\centering
\includegraphics[width=0.9\columnwidth]{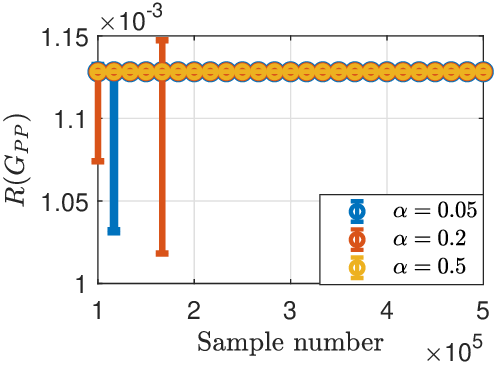}
    \caption{Estimated transmission rate with the corresponding CI at $\alpha = 0.05,0.2,0.5$ at different sample numbers for targeted error probability of $10^{-5}$.}
   \label{fig:ratesampleci}
\end{figure}
\vspace{-0.5cm}
\section{Conclusions}
\label{sec:conclusions}
In this paper, we propose a novel framework based on the multivariate extreme value theory with the goal of determining an optimum transmission rate for a system operating in MIMO-URC, and incorporating a confidence interval for the estimated BGPD parameters and the corresponding optimum transmission rate to enable the usage of MEVT in real-time communication. The framework includes the modeling of the inter-relationship of the tail distributions of the multiple channels by using the BGPD based on the Poisson point process approach, determination of the optimum transmission rate by using the estimated BGPD model, incorporation of the  confidence intervals for the estimated transmission rate due to the availability of the restricted amount of data, and then assessment of the system reliability by utilizing the outage probability metric. The proposed MEVT-based rate selection framework achieves a higher transmission rate, up to about $10^3$, than the traditional methods based on the extrapolation of average statistic channel models.
Incorporating CI into the rate selection estimate can lower the required number of samples, at least by $10$-fold, during the training phase necessary to achieve a specific level of reliability. This reduction in required samples subsequently decreases the sample complexity and contributes to a more efficient rate estimation. 
It has been demonstrated that calculating GPD parameters using limited samples leads to a wider confidence interval in rate estimation, implying a greater degree of uncertainty. Nevertheless, as additional data become accessible, the confidence limits become narrower because of the diminishing level of uncertainty.
Moreover, based on the proposed methodology, we demonstrate the trade-off between the uncertainty in the rate selection and the cost of data collection. 
Since a large amount of data is required for deriving the statistics, our future endeavors encompass the expansion of our research through the implementation of transfer learning techniques, specifically employing knowledge-assisted training. This approach entails the development of a digital twin model that accurately represents a real-time network, incorporating essential elements such as network topology, channel characteristics, and queueing models for offline training purposes. Subsequently, this model is refined and optimized using a smaller dataset in the live operating environment.



\balance 

\ifCLASSOPTIONcaptionsoff
  \newpage
\fi

\bibliographystyle{ieeetr}
\bibliography{MultivariateEVT_RateSelection_CI.bib}

\begin{thebibliography}{10}

\bibitem{interface_01}
P.~Popovski, J.~J. Nielsen, C.~Stefanovic, E.~De~Carvalho, E.~Strom, K.~F. Trillingsgaard, A.-S. Bana, D.~M. Kim, R.~Kotaba, J.~Park, {\em et~al.}, ``Wireless access for ultra-reliable low-latency communication: Principles and building blocks,'' {\em IEEE Network}, vol.~32, no.~2, pp.~16--23, 2018.

\bibitem{5G_01}
C.~{Gustafson}, K.~{Haneda}, S.~{Wyne}, and F.~{Tufvesson}, ``On mm-wave multipath clustering and channel modeling,'' {\em IEEE Transactions on Antennas and Propagation}, vol.~62, no.~3, pp.~1445--1455, 2014.

\bibitem{interface_02}
P.~Popovski, {\v{C}}.~Stefanovi{\'c}, J.~J. Nielsen, E.~De~Carvalho, M.~Angjelichinoski, K.~F. Trillingsgaard, and A.-S. Bana, ``Wireless access in ultra-reliable low-latency communication ({URLLC}),'' {\em IEEE Transactions on Communications}, vol.~67, no.~8, pp.~5783--5801, 2019.

\bibitem{urllc_01}
M.~Bennis, M.~Debbah, and H.~V. Poor, ``Ultra reliable and low-latency wireless communication: Tail, risk, and scale,'' {\em Proceedings of the IEEE}, vol.~106, no.~10, pp.~1834--1853, 2018.

\bibitem{urllc_02}
P.~C.~F. {Eggers}, M.~{Angjelichinoski}, and P.~{Popovski}, ``Wireless channel modeling perspectives for ultra-reliable communications,'' {\em IEEE Transactions on Wireless Communications}, vol.~18, no.~4, pp.~2229--2243, 2019.

\bibitem{urllc_05}
M.~Angjelichinoski, K.~F. Trillingsgaard, and P.~Popovski, ``A statistical learning approach to ultra-reliable low latency communication,'' {\em IEEE Transactions on Communications}, vol.~67, no.~7, pp.~5153--5166, 2019.

\bibitem{timediv}
M.~Serror, C.~Dombrowski, K.~Wehrle, and J.~Gross, ``Channel coding versus cooperative arq: Reducing outage probability in ultra-low latency wireless communications,'' in {\em 2015 IEEE Globecom Workshops (GC Wkshps)}, pp.~1--6, Dec 2015.

\bibitem{frequencydiv}
G.~Pocovi, B.~Soret, M.~Lauridsen, K.~I. Pedersen, and P.~Mogensen, ``Signal quality outage analysis for ultra-reliable communications in cellular networks,'' in {\em 2015 IEEE Globecom Workshops (GC Wkshps)}, pp.~1--6, Dec 2015.

\bibitem{spacediv}
G.~Pocovi, B.~Soret, M.~Lauridsen, K.~I. Pedersen, and P.~Mogensen, ``Signal quality outage analysis for ultra-reliable communications in cellular networks,'' in {\em 2015 IEEE Globecom Workshops (GC Wkshps)}, pp.~1--6, Dec 2015.

\bibitem{interferencediv}
J.~J. Nielsen, R.~Liu, and P.~Popovski, ``Ultra-reliable low latency communication using interface diversity,'' {\em IEEE Transactions on Communications}, vol.~66, pp.~1322--1334, March 2018.

\bibitem{urllc_diversity_01}
S.~R. Khosravirad, H.~Viswanathan, and W.~Yu, ``Exploiting diversity for ultra-reliable and low-latency wireless control,'' {\em IEEE Transactions on Wireless Communications}, vol.~20, no.~1, pp.~316--331, 2020.

\bibitem{urllc_08}
V.~N. Swamy, P.~Rigge, G.~Ranade, B.~Nikoli{\'c}, and A.~Sahai, ``Wireless channel dynamics and robustness for ultra-reliable low-latency communications,'' {\em IEEE Journal on Selected Areas in Communications}, vol.~37, no.~4, pp.~705--720, 2019.

\bibitem{reliability_01}
S.~Samarakoon, M.~Bennis, W.~Saad, and M.~Debbah, ``Predictive ultra-reliable communication: A survival analysis perspective,'' {\em IEEE Communications Letters}, vol.~25, no.~4, pp.~1221--1225, 2020.

\bibitem{Mehrnia_twcBGPD}
N.~Mehrnia and S.~Coleri, ``Multivariate extreme value theory based channel modeling for ultra-reliable communications,'' {\em IEEE Transactions on Wireless Communications}, pp.~1--1, 2023.

\bibitem{urllc_07}
V.~N. Swamy, P.~Rigge, G.~Ranade, B.~Nikolic, and A.~Sahai, ``Predicting wireless channels for ultra-reliable low-latency communications,'' in {\em 2018 IEEE International Symposium on Information Theory (ISIT)}, pp.~2609--2613, IEEE, 2018.

\bibitem{reliability_03}
F.~Pase, M.~Giordani, and M.~Zorzi, ``On the convergence time of federated learning over wireless networks under imperfect csi,'' {\em arXiv preprint arXiv:2104.00331}, 2021.

\bibitem{MehrniaTVTRate}
N.~Mehrnia and S.~Coleri, ``Extreme value theory based rate selection for ultra-reliable communications,'' {\em IEEE Transactions on Vehicular Technology}, pp.~1--1, 2022.

\bibitem{MehrniaTWC}
N.~Mehrnia and S.~Coleri, ``Wireless channel modeling based on extreme value theory for ultra-reliable communications,'' {\em IEEE Transactions on Wireless Communications}, pp.~1--1, 2021.

\bibitem{MehrniaTVT}
N.~Mehrnia and S.~Coleri, ``Non-stationary wireless channel modeling approach based on extreme value theory for ultra-reliable communications,'' {\em IEEE Transactions on Vehicular Technology}, vol.~70, no.~8, pp.~8264--8268, 2021.

\bibitem{confidenceinterval_03}
W.~Zhang, M.~Derakhshani, and S.~Lambotharan, ``Non-parametric statistical learning for urllc transmission rate control,'' in {\em ICC 2021-IEEE International Conference on Communications}, pp.~1--6, IEEE, 2021.

\bibitem{MehrniaEucnc}
N.~Mehrnia and S.~Coleri, ``Incorporation of confidence interval into rate selection based on the extreme value theory for ultra-reliable communications,'' in {\em European Conference on Networks and Communications (EuCNC) and the 6G Summit}, pp.~1--6, IEEE, 2022.

\bibitem{evt_04}
S.~Coles, J.~Bawa, L.~Trenner, and P.~Dorazio, {\em An introduction to statistical modeling of extreme values}, vol.~208.
\newblock Springer, 2001.

\bibitem{urllc_noise}
G.~Durisi, T.~Koch, and P.~Popovski, ``Toward massive, ultrareliable, and low-latency wireless communication with short packets,'' {\em Proceedings of the IEEE}, vol.~104, no.~9, pp.~1711--1726, 2016.

\bibitem{corcoef_01}
J.~Shi, D.~Anzai, and J.~Wang, ``Channel modeling and performance analysis of diversity reception for implant uwb wireless link,'' {\em IEICE transactions on communications}, vol.~95, no.~10, pp.~3197--3205, 2012.

\bibitem{pickands_02}
D.~Cooley, R.~A. Davis, and P.~Naveau, ``The pairwise beta distribution: A flexible parametric multivariate model for extremes,'' {\em Journal of Multivariate Analysis}, vol.~101, no.~9, pp.~2103--2117, 2010.

\bibitem{newton}
A.~Gil, J.~Segura, and N.~M. Temme, {\em Numerical methods for special functions}.
\newblock SIAM, 2007.

\bibitem{ci_10}
M.-T. Puth, M.~Neuh{\"a}user, and G.~D. Ruxton, ``On the variety of methods for calculating confidence intervals by bootstrapping,'' {\em Journal of Animal Ecology}, vol.~84, no.~4, pp.~892--897, 2015.

\bibitem{bootci_02}
B.~Gr{\"u}n and T.~Miljkovic, ``The automated bias-corrected and accelerated bootstrap confidence intervals for risk measures,'' {\em North American Actuarial Journal}, pp.~1--20, 2022.

\bibitem{ci_01}
M.~H. Quenouille, ``Approximate tests of correlation in time-series 3,'' in {\em Mathematical Proceedings of the Cambridge Philosophical Society}, vol.~45, pp.~483--484, Cambridge University Press, 1949.

\bibitem{ci_02}
R.~J. Tibshirani and B.~Efron, ``An introduction to the bootstrap,'' {\em Monographs on statistics and applied probability}, vol.~57, pp.~1--436, 1993.

\bibitem{vehicular_01}
C.~U. Bas and S.~Coleri~Ergen, ``Ultra-wideband channel model for intra-vehicular wireless sensor networks beneath the chassis: From statistical model to simulations,'' {\em IEEE Transactions on Vehicular Technology}, vol.~62, no.~1, pp.~14--25, 2012.

\bibitem{vehicular_02}
U.~Demir, C.~U. Bas, and S.~Coleri~Ergen, ``Engine compartment {UWB} channel model for intra-vehicular wireless sensor networks,'' {\em IEEE Transactions on Vehicular Technology}, vol.~63, no.~6, pp.~2497--2505, 2013.

\bibitem{extrpltd_01}
P.~Bortot, S.~Coles, and J.~Tawn, ``The multivariate gaussian tail model: An application to oceanographic data,'' {\em Journal of the Royal Statistical Society: Series C (Applied Statistics)}, vol.~49, no.~1, pp.~31--049, 2000.

\bibitem{demir2013engine}
U.~Demir, C.~U. Bas, and S.~C. Ergen, ``Engine compartment uwb channel model for intravehicular wireless sensor networks,'' {\em IEEE Transactions on Vehicular Technology}, vol.~63, no.~6, pp.~2497--2505, 2013.

\end{thebibliography}

\end{document}